\documentclass[12pt]{llncs}
\usepackage{amsfonts}
\usepackage{psfrag}

\usepackage[utf8]{inputenc} 
\usepackage{amsmath}
\usepackage{amssymb}
\usepackage{tikz}

\title{On the Nearest Neighbor Rule for the\\ Metric Traveling Salesman Problem}
\author{Stefan Hougardy and Mirko Wilde}
\institute{Research Institute for Discrete Mathematics,
           University of Bonn\\
           Lenn\'estr.~2, 53113 Bonn, Germany\\[5mm]
           \today}
\begin{document}
\maketitle 

\begin{abstract}
\small
We present a very simple family of traveling salesman instances with $n$ cities 
where the nearest neighbor rule may produce a tour that is $\Theta(\log n)$ times longer
than an optimum solution. Our family works for the graphic, the
euclidean, and the rectilinear traveling salesman problem at the same time. 
It improves the so far best known lower bound in the euclidean case and 
proves for the first time a lower bound in the rectilinear case. 
\end{abstract}

{\small\textbf{keywords:} traveling salesman problem; nearest neighbor rule; approximation algorithm}

\section{Introduction}

Given $n$ cities with their pairwise distances $d_{i,j}$ the \emph{traveling salesman problem} (TSP) 
asks for a shortest tour that visits each city exactly once.
This problem is known to be NP-hard~\cite{Kar1972} and therefore much effort has been spent
to design efficient heuristics that are able to find good tours.
A heuristic $A$ for the traveling salesman problem is said to have \emph{approximation ratio}
$c$ if for every TSP instance it finds a tour that is at most 
$c$ times longer than a shortest tour. 
We will consider here only \emph{metric} TSP instances, i.e., TSP instances where
the distances between the $n$ cities satisfy the triangle inequality $d_{i,j}\le d_{i,k} + d_{k,j}$
for all $1\le i,j,k\le n$. Well studied special cases of the metric TSP are the
\emph{euclidean} and the \emph{rectilinear} TSP. In these instances the cities are points in the plane
and the distance between two cities is defined as the euclidean respectively
rectilinear distance. A third example of metric TSP instances are the \emph{graphic} 
TSP instances. Such an instance is obtained from an (unweighted, undirected) 
connected graph $G$ which has as vertices
all the cities. The distance between two cities is then defined as the length of a shortest path
in $G$ that connects the two cities.

One of the most natural heuristics for the TSP is the \emph{nearest neighbor rule} (NNR)~\cite{Flo1956}. 
This rule grows \emph{partial tours} of increasing size where a partial tour is an ordered subset
of the cities. The nearest neighbor rule starts with a partial tour consisting of a single city $x_1$.
If the nearest neighbor rule has constructed a partial tour $(x_1, x_2, \ldots, x_k)$ then it extends
this partial tour by a city $x_{k+1}$ that has smallest distance to $x_k$ and is not yet contained
in the partial tour. Ties are broken arbitrarily. 
A partial tour that contains all cities yields a TSP tour by going from the last 
city to the first city in the partial tour. Any tour that can be obtained this way is called an 
NNR tour.

\subsection{Known Results}

Rosenkrantz, Stearns, and Lewis~\cite{RSL1977} proved that on an $n$ city metric TSP instance
the nearest neighbor rule has approximation ratio at most
$\frac12\lceil\log n\rceil +\frac12$, where throughout the paper $\log$ denotes the logarithm with base 2.
They also constructed a family of metric TSP instances that show
a lower bound of $\frac13\log n$ for the approximation ratio of the 
nearest neighbor rule.
Johnson and Papadimitriou~\cite{JP1985} presented a simplified construction that yields a lower
bound of $\frac16\log n$. 
Hurkens and Woeginger~\cite{HW2004} constructed a simple family of graphic TSP instances
that proves a lower bound of $\frac14\log n$. Moreover they present in the same paper
a simple construction of a family of euclidean TSP instances that proves 
that the nearest neighbor rule has approximation
ratio at least $\left(\sqrt 3 -\frac32\right) \left(\log n -2\right)$ where 
$\sqrt 3 -\frac32 \le 0.232$.  
For the graphic TSP Pritchard~\cite{Pri2007} presents a more complicated construction
that shows that for each $\epsilon > 0$ 
the approximation ratio of the nearest neighbor rule on graphic TSP instances
is at least $\frac1{2+\epsilon} \log n$.

\subsection{Our Contribution}

We will present in the next section a very simple construction that proves a lower bound of
$\frac14\log n$ for the graphic, euclidean, and rectilinear TSP at the same time. 
Our construction proves for the first time a lower bound on the approximation ratio
of the nearest neighbor rule for rectilinear TSP instances. Moreover, we improve
the so far best known lower bound of~\cite{HW2004} for the approximation ratio of the nearest 
neighbor rule for euclidean TSP instances.

\section{The Construction of Bad Instances}\label{sec:construction}

We will now describe our construction of a family of metric TSP instances $G_k$ on which 
the nearest neighbor rule yields tours that are much longer than optimum tours.
As the set of cities $V_k$ we take the points of a $2\times (8\cdot 2^k-3)$ subgrid of $\mathbb{Z}^2$.
Thus we have $|V_k|= 16 \cdot 2^k-6$.
Let $G_k$ be any TSP-instance defined on the cities $V_k$ that satisfies the following conditions:
\begin{itemize}
\item[(i)] if two cities have the same x-coordinate or the same y-coordinate
           their distance is the euclidean distance between the two cities.
\item[(ii)] if two cities have different x-coordinate and  different y-coordinate
           then their distance is at least as large as the absolute difference between
           their x-coordinates.       
\end{itemize}

Note that if we choose as $G_k$ the euclidean or the rectilinear TSP instance on $V_k$ 
then conditions (i) and (ii) are satisfied.
We can define a graph on $V_k$ by adding an edge between each pair of cities
at distance 1. The graphic TSP that is induced by this graph is exactly the rectilinear TSP.
The graph for $G_0$ is shown in Figure~\ref{fig:G0}. We label the lower left vertex
in $V_k$ as $l_k$ and the top middle vertex in $V_k$ as $m_k$. 

\begin{figure}[ht]
\centering
\begin{tikzpicture}[scale=1]
\draw[step = 1cm] (0,0) grid (4,1); 
\foreach \x in {0,1,2,3,4}
	\foreach \y in {0,1}
		\fill (\x, \y) circle(1mm);
\draw[blue,->] (0.1,0.1) -- (0.1, 0.85);
\draw[blue,->] (0.1,0.9) -- (0.9,0.9);
\draw[blue,->] (0.9,0.85) -- (0.9,0.1);
\draw[blue,->] (1.1,0.1) -- (1.9,0.1);
\draw[blue,->] (2.1,0.1) -- (2.9,0.1);
\draw[blue,->] (3.1,0.1) -- (3.9,0.1);
\draw[blue,->] (3.9,0.1) -- (3.9, 0.85);
\draw[blue,->] (3.9,0.9) -- (3.1,0.9);
\draw[blue,->] (2.9,0.9) -- (2.1,0.9);

\draw (0,0) node[anchor = east] {$l_0$};
\draw (2,1) node[anchor = south] {$m_0$};
\end{tikzpicture}
\caption{The graph defining the graphic TSP $G_0$ together with a partial NNR tour that 
connects $l_0$ with $m_0$.}
\label{fig:G0}
\end{figure}

Our construction of an NNR tour in $G_k$ will only make use of the properties (i) and (ii) of $G_k$. 
Thus with the same proof we get a result for euclidean, rectilinear and graphic TSP instances.
We will prove by induction on $k$ that the nearest neighbor rule can find a rather long tour in $G_k$. 
For this we need to prove the following slightly more general result which is 
similar to Lemma~1 in~\cite{HW2004}.

\begin{lemma}
\label{lemma}
Let the cities of $G_k$ be embedded into $G_m$ with $m > k$. Then there exists a partial NNR tour
in $G_m$ that 
\begin{itemize}
\item[\rm(a)] visits exactly the cities in $G_k$,
\item[\rm(b)] starts in $l_k$ and ends in $m_k$, and 
\item[\rm(c)] has length exactly $(12+4k)\cdot 2^k-3$.
\end{itemize}
\end{lemma}

\begin{proof}
We use induction on $k$ to prove the statement. For $k=0$ a
partial NNR tour of length $12\cdot 2^0-3=9$ that satisfies (a), (b), and (c)
is shown in Figure~\ref{fig:G0}.
Now assume we already have defined a partial NNR tour for $G_k$. 
Then we define a partial NNR tour for $G_{k+1}$ recursively as follows.
As $|V_{k+1}| = 16 \cdot 2^{k+1}-6 = 2\cdot(16 \cdot 2^k-6) + 6 = 2\cdot |V_k| + 6$
we can think of $G_{k+1}$ to be the disjoint union of two copies $G_k'$ and $G_k''$ of $G_k$
separated by a $2\times 3$ grid. This is shown in Figure~\ref{fig:Gk}.

\begin{figure}[ht]
\centering
\begin{tikzpicture}
\draw (0,0) -- (0,1) -- (4,1) -- (4,0) -- (0,0) ;
\draw[step = 1cm] (5,0) grid (7,1); 
\draw (8,0) -- (8,1) -- (12,1) -- (12,0) -- (8,0) ;
\fill (0,0) circle(0.1); \draw (0,0) node[anchor = north] {$l_{k+1}=l_k'$};
\fill (2,1) circle(0.1); \draw (2,1) node[anchor = south] {$m_k'$};
\foreach \x in {5,6,7} \fill (\x,0) circle(0.1);
\foreach \x in {5,6,7} \fill (\x,1) circle(0.1);
\fill (8,0) circle(0.1); \draw (8,0) node[anchor = north] {$l_k''$};
\fill (10,1) circle(0.1); \draw (10,1) node[anchor = south] {$m_k''$};
\draw (6,1) node[anchor = south] {$m_{k+1}$};
\draw[dash pattern = on 2pt off 2pt,blue,->] (0.1,0.05) -- (1.9,0.95);
\draw[blue,->] (2.1,0.95) .. controls (4.5,0.75) .. (4.9,0.95);
\draw[blue,->] (5.1,0.9) -- (5.1, 0.15);
\foreach \x in {5,6,7} \draw[blue,->] (\x.1,0.1) -- (\x.9, 0.1);
\draw[dash pattern = on 2pt off 2pt,blue,->] (8.1,0.05) -- (9.9,0.95);
\draw[blue,->] (9.9,1.05) .. controls (7.5,1.25) .. (7.1,1.05);
\draw[blue,->] (6.9,0.9) -- (6.1, 0.9);
\draw (2,0.4) node {$G'_k$};
\draw (10,0.4) node {$G''_k$};
\end{tikzpicture}
\caption{The recursive construction of a partial NNR tour for the instance $G_{k+1}$. The dashed lines 
indicate partial NNR tours in $G_k'$ and $G_k''$. }
\label{fig:Gk} 
\end{figure}
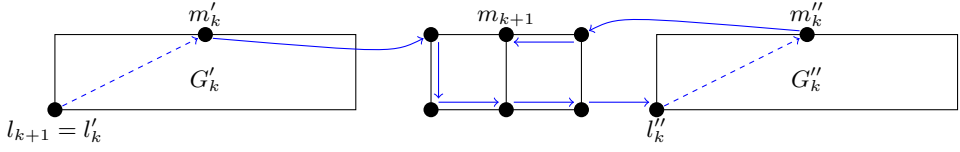

Now we can construct a partial NNR tour for $G_{k+1}$ as follows. Start in vertex $l_k'$ and follow
the partial NNR tour in $G'_k$ that ends in vertex $m_k'$. The leftmost top vertex of the $2\times 3$-grid
is now a closest neighbor of  $m_k'$ that has not been visited so far. Go to this vertex, 
visit some of the vertices of the $2\times 3$-grid as indicated in Figure~\ref{fig:Gk}, and
then go to vertex $l_k''$. From this vertex follow the partial NNR tour in $G_k''$ which ends in 
vertex $m_k''$. A nearest neighbor for this vertex now is the rightmost top vertex of
the $2\times3$-grid. Continue with this vertex and go to the left to reach vertex $m_{k+1}$. 
The partial NNR tour constructed this way obviously satisfies conditions (a) and (b). Moreover,
this partial NNR tour is also a partial NNR tour when $G_{k+1}$ is embedded into some $G_l$ for $l>k+1$.

The length of this partial NNR tour is twice the length of the partial NNR tour in $G_k$
plus five edges of length 1 plus two edges of length $\frac12\left(\frac12\cdot |V_k|+1\right)$.
Thus we get a total length of 
$$2 \left((12+4k)\cdot 2^k-3\right) + 5 + 8\cdot 2^k -2 ~=~ (12+4(k+1))\cdot 2^{k+1}-3$$
for the partial NNR tour constructed in $G_{k+1}$. This proves condition (c).
\end{proof}

\begin{theorem}
\label{thm:main}
On graphic, euclidean, and rectilinear TSP instances with $n$ cities
the approximation ratio of the nearest neighbor rule is no better than $\frac14\cdot \log n - 1$.
\end{theorem}

\begin{proof}
The instance $G_k$ defined above has $n:=16 \cdot 2^k-6$ cities and an optimum TSP tour 
in $G_k$ has length $n$. As shown in Lemma~\ref{lemma} there exists a partial NNR tour in $G_k$ 
of length at least $(12+4k)\cdot 2^k-3$. Thus the approximation ratio of the nearest neighbor rule
is no better than 
$$\frac {(12+4k)\cdot 2^k-3}{16 \cdot 2^k-6} ~\ge~ \frac{12+4k}{16}~=~\frac{3+k}{4}~=~\frac{3+\log\left(\frac{n+6}{16}\right)}{4}~\ge~\frac14\left(\log n - 1\right).$$
\end{proof}

\section{Comments}

Conditions (i) and (ii) in Section~\ref{sec:construction} are satisfied whenever
the distances in $G_k$ are defined by an $L^p$-norm. Thus Theorem~\ref{thm:main} not only holds
for the $L^2$- and the $L^1$-norm but for all $L^p$-norms.      
As already noted in~\cite{JP1985} and~\cite{HW2004}
one can make the NNR tour unique in the euclidean and rectilinear case (and more generally in the $L^p$-case)
by moving all cities by some small amount.
The $\Theta(\log n)$ lower bound of our family is independent of the city in which the nearest neighbor rule starts.

\bibliographystyle{plain}
\bibliography{NNR}

\end{document}